\newcommand{\CLASSINPUTbaselinestretch}{1.2}
\documentclass[conference,twocolumn]{IEEEtran}

\usepackage{geometry}
\usepackage[english]{babel}
\usepackage{enumitem}
\usepackage[latin1]{inputenc}
\usepackage{enumerate}
\usepackage{color}
\usepackage[T1]{fontenc}
\usepackage{subfigure}
\usepackage{dsfont}
\usepackage[final]{graphicx}
\usepackage[T1]{fontenc}
\usepackage{amsmath}
\usepackage{mathtools, cuted}
\usepackage{amsthm}
\usepackage{amstext}
\usepackage{amssymb}
\usepackage{mathrsfs}
\usepackage{cite}
\usepackage{mathtools}
\usepackage{tikz}

\usetikzlibrary{arrows,positioning, shapes}
\usepackage{float}

\usepackage{setspace}
\singlespace
\usepackage{hyperref}
\def\R{\mathbb{R}}

\def\eps{\varepsilon}

\def\hY{\hat{Y}}

\def\A{{\mathcal A}}

\def\E{{\mathbb E}}

\def\T{{\mathcal T}}
\def\X{{\mathcal X}}
\def\Y{{\mathcal Y}}

\def\cR{{\mathcal R}}
\def\Z{{\mathcal Z}}
\def\U{{\mathcal U}}
\def\V{{\mathcal V}}

\usepackage[font=small,labelsep=space]{caption}
\DeclareCaptionLabelSeparator{dot}{.~}
\newcounter{example}

\newtheorem{theorem}{Theorem}
\newtheorem{corollary}{Corollary}

\newtheorem{lemma}{Lemma}
\theoremstyle{remark}
\newtheorem{remark}{Remark}

\newcommand{\markov}{\mathrel\multimap\joinrel\mathrel-%
\mspace{-9mu}\joinrel\mathrel-}

\usepackage{times}
\usepackage{tikz}
\usepackage{amsmath}
\usepackage{verbatim}
\usetikzlibrary{arrows,shapes}
\tikzstyle{RectObject}=[rectangle,fill=white,draw,line width=0.2mm]
\tikzstyle{line}=[draw]
\tikzstyle{arrow}=[draw, -latex]
\usetikzlibrary{decorations.pathmorphing}
\usetikzlibrary{calc,shapes, positioning}
\usepackage{graphicx}
\usepackage{caption}
\usetikzlibrary{shapes.geometric}

\IEEEoverridecommandlockouts

\allowdisplaybreaks

\begin{document}

\title{\vspace{5.5mm}Lossless Secure Source Coding: ~~~~~~~~~Yamamoto's Setting}
%\author{\IEEEauthorblockN{Shahab Asoodeh}
\author{\IEEEauthorblockN{Shahab Asoodeh\thanks{This work was supported in part by NSERC of Canada.}, Fady Alajaji, and Tam\'{a}s Linder}
    \IEEEauthorblockA{Department of Mathematics and Statistics, Queen's University
    %\\ asoodehshahab@mast.queensu.ca  }}
    \\\{asoodehshahab, fady, linder\}@mast.queensu.ca}}
\restoregeometry
\maketitle

\begin{abstract}
Given a private source of information, $X^n$ and a public correlated source, $Y^n$, we study the problem of encoding the two-dimensional source $(X^n, Y^n)$ into an index $J$ such that a remote party, knowing $J$ and some external side information $Z^n$, can losslessly recover $Y^n$ while any eavesdropper knowing $J$ and possibly a correlated side information $E^n$ can retrieve very little information about $X^n$. We give general converse results for the amount of information about $X^n$ that might be leaked in such systems and and also achievability results that are optimal in some special cases.
\end{abstract}

\section{Introduction}
Information-theoretic secrecy models concern a tradeoff between utility and privacy. Given a source $Y^n$, the goal is to transmit this source securely and reliably over a noiseless public channel which might be perfectly observed by a passive adversary. The utility is defined as the accuracy in the recovering of $Y^n$ by a remote receiver and the privacy is defined as the uncertainty of the source given the message sent over the channel. However, in some cases, it may be desirable to define utility and privacy for two different sources, that is, we want the receiver to know $Y^n$ with some level of accuracy while revealing very little information about a correlated source $X^n$, which we refer to as the private source.

To motivate this setting, consider the following example. Suppose $Y$ denotes an attribute of a bank customer that a trusted advertising company would like to target and $X$ denotes another, more sensitive, attribute of the customer. The bank has database $(X^n, Y^n)$ corresponding to $n$ different users. The company pays the bank to receive $Y^n$ as accurately as possible. However, some governing laws prohibit the database $X^n$ from being revealed too extensively over public communication channels. Consequently, the data given to the company must be chosen so that at most a prescribed amount of information is revealed about $X^n$ over the communication channel while the recovery of $Y^n$ by the company satisfies some level of quality.

Inspired by Yamamoto \cite{yamamotoequivocationdistortion} where a lossy source coding problem is studied under a privacy constraint, we consider a secure lossless source coding model in which an encoder (Alice) encodes a two-dimensional source $(X^n,Y^n)$ such that the receiver (Bob) is able to reconstruct $Y^n$ correctly with high probability and the leakage of information (the information obtained by an eavesdropper, Eve) about $X^n$ is no more than $\Delta\geq0$. It is clear that no non-trivial level of privacy can be obtained if no side information is available to Bob. Hence, we assume Bob has access to some correlated side information and after observing the channel output wants to recover $Y^n$ with asymptotically vanishing error probability. We study this problem in terms of the compression rate and also the information leakage about $X^n$ (or equivalently the equivocation between the compressed and the private data). We give converse results for different cases including when Bob has coded or uncoded side information, when Eve has uncoded side information, or when the private source, $X^n$, is hidden even from Alice.

When $X=Y$, the problem we consider here reduces to a well-known model which has been extensively studied, for example see \cite{Secure_distributed_Vinod, Secure_lossless_compression, Secrure_lossless_Gunduz_ISIT, Secure_Source_Piantanida, Secure_source_Tandon}. In particular, Prabhakaran and Ramchandran \cite{Secure_distributed_Vinod} considered a similar secure lossless setting with $X=Y$ and Bob and Eve having correlated uncoded side information. They focused on the best achievable information leakage rate when the public channel has not rate limit. G\"{u}nd\"{u}z et al.\ \cite{Secure_lossless_compression}, \cite{Secrure_lossless_Gunduz_ISIT} gave converse and achievability bounds for a similar setting for both compression rate and information leakage which do not necessarily match. Tandon et al.\ \cite{Secure_source_Tandon} considered a simpler case in which Eve has no side information, gave a single letter characterization of the optimal rates, and information leakage and showed that a simple coding scheme based on binning, similar to the one proposed by Wyner in \cite{Wyner_source_coidng_wit_SI}, is indeed optimal with and without the privacy constraint. Our results recover all these results in the special case of $X=Y$.

The rest of this paper is organized as follows. In Section II, we formally define our problem and state an outer bound which is our main result. In Section III, we consider a more general model in which Eve has side information and  present another outer bound. We then present a coding scheme which is shown to be optimal in some special cases. We complete the paper with some concluding remarks in Section IV.

\section {Yamamoto's Lossless Source Coding: Coded Side Information at Bob}
Yamamoto \cite{yamamotoequivocationdistortion} considered a lossy source coding scheme with a privacy constraint at the legitimate decoder. This is contrasted with the typical information-theoretic secrecy models in which the privacy is defined as the uncertainty of the source against a passive eavesdropper. In this model, having observed $(X^n, Y^n)$, the encoder $\varphi:\X^n\times \Y^n\to\{1, 2, \dots, 2^{nR}\}$, transmits a message to the decoder, $\psi:\{1, 2, \dots, 2^{nR}\}\to\hat{\Y}^n$, which is required to recover $Y^n$ within some distortion $D$ while revealing little information about $X^n$. More precisely, for a given distortion measure $d:\Y\times \hat{\Y}\to \R_+$, we require $\frac{1}{n}\sum \E[d(Y_i, \hY_i)]\leq D$ while the normalized uncertainty about $X^n$ at the decoder is lower-bounded, i.e., $\frac{1}{n}H(X^n|\varphi(X^n, Y^n))\geq E$ for a non-negative $E\leq H(X)$. This requirement is different from the privacy constraint usually considered in information-theoretic secrecy (e.g., \cite{Secure_lossless_compression}, \cite{Secure_lossy_coding}, \cite{Secure_source_Tandon}, and \cite{Secure_Source_Piantanida}), in that here the utility and privacy are measured with respect to two different sources $Y$ and $X$, respectively. In this sense, $X$ and $Y$ correspond to the private and public sources, respectively. The correlation between $X$ and $Y$ makes the utility and privacy constraints contradicting.
\newcounter{tempequationcounter}

We study a similar model as Yamamoto's but for \emph{lossless} compression. Clearly, if no side information is available to the decoder, then the eavesdropper can obtain as much information about $X^n$ as the legitimate decoder and hence only trivial levels of privacy can be achieved when lossless compression of $Y$ is required. We, therefore, assume that side information is provided at the decoder, as depicted in Fig.~\ref{Figue:Yamamoto_Lossless}.
\tikzstyle{int}=[draw, fill=blue!20, minimum size=2em]
\tikzstyle{init} = [pin edge={to-,thin,black}]
\tikzstyle{init_to} = [pin edge={-to,thin,black}]
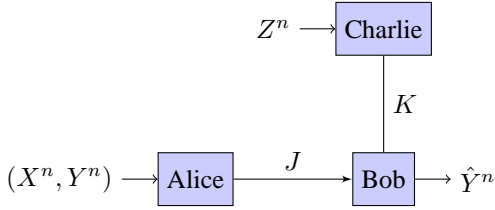
\begin{figure}[h]
\centering
\begin{tikzpicture}[node distance=2.5cm,auto,>=latex']
    \node [int, pin={[init]left:$(X^n,Y^n)$}] (a) [node distance=1.5cm] {Alice};
    \node (b) [left of=a,node distance=2cm, coordinate] {a};
    \node [coordinate] (end) [right of=b, node distance=2cm]{};
    \node (b) [left of=a,node distance=2cm, coordinate] {a};
    \node [int, pin={[init_to]right:$\hat{Y}^n$}] (c) [right of=a] {Bob};
    \node [int, pin={[init]left:$Z^n$}] (d) [above of=c, node distance=2cm] {Charlie};
    \node [coordinate] (end) [right of=c, node distance=2cm]{};
    \path[->] (a) edge node {$J$} (c);
    \path[-] (d) edge node {$K$} (c);
    %\draw[->] (c) edge node {$\hat{Y}^n$} (end) ;
\end{tikzpicture}
\caption{Yamamoto's lossless source coding.}
\label{Figue:Yamamoto_Lossless}
\end{figure}

A $(2^{nR_A}, 2^{nR_C}, n)$ code for private lossless compression in this setup is composed of two encoding functions at Alice and Charlie, respectively, $f_A:\X^n\times \Y^n\to \{1, 2, \dots, 2^{nR_A}\}$ and $f_C:\Z^n\to \{1, 2, \dots, 2^{nR_C}\}$, and a decoder at Bob, $f_B:\{1, 2, \dots, 2^{nR_A}\}\times \{1, 2, \dots, 2^{nR_C}\}\to \hat{\Y}^n$, where $(X^n, Y^n, Z^n)$ are $n$ independent and identically distributed (i.i.d.) copies of $(X,Y,Z)$ with joint distribution $P(x,y,z)$. We assume that both encoders communicate to Bob over noiseless channels; however, the channel between Alice and Bob is subject to eavesdropping and hence a passive party can have access to the message $J$ transmitted over this channel. A triple $(R_A, R_C, \Delta)\in\R^3_+$ is said to be achievable if for any $\eps>0$, there exists a $(2^{nR_A}, 2^{nR_C}, n)$ code for $n$ large enough such that
\begin{eqnarray}
% \nonumber to remove numbering (before each equation)
  \Pr(f_B(J, K)\neq Y^n) &<& \eps, \label{error-probability-constraint}\\
  \frac{1}{n} H(X^n|J)&\geq& \Delta-\eps \label{privacy-constraint},
\end{eqnarray}
where $J:=f_A(X^n, Y^n)$ and $K:=f_C(Z^n)$. We denote the set of all achievable triples $(R_A, R_C, \Delta)$ by $\cR$. One special case of interest is when $J$ contains absolutely no information about the private source, that is, when $J$ is independent of $X^n$, which is called perfect privacy.

We note that for a special case of $X=Y$, inner and outer bounds on the achievable region were initially presented in \cite[Theorem 3.1]{Secrure_lossless_Gunduz_ISIT}, although these bounds do no match in general. Tight bounds were then given in
\cite[Theorem 1]{Secure_source_Tandon} whose achievability resembles the binning scheme proposed by Wyner \cite{Wyner_source_coidng_wit_SI} for standard source coding with coded side information at the decoder. This therefore shows that the privacy constraint  \eqref{privacy-constraint} does not change the optimal scheme.
\begin{theorem}\label{Theorem_lossless_Yamamoto}
For any achievable triple $(R_A, R_C, \Delta)\in \cR$ we have
\begin{eqnarray*}
% \nonumber to remove numbering (before each equation)
  R_A &\geq& H(Y|V), \\
  R_C &\geq & I(Z;V), \\
  \Delta &\leq & I(X,Y;V)+ H(X|U)-H(Y|U),
\end{eqnarray*}
for some  auxiliary random variables $V\in \V$ and $U\in \U$ such that $P(x,y,z,u,v)=P(x,y,z)P(v|z)P(u|x,y)$ with  $|\U|\leq |\X|\times|\Y|+1$ and $|\V|\leq |\Z|+2$.
 \end{theorem}
\begin{proof}
First note that Bob is required to reconstruct $Y^n$ losslessly given $J$ and $K$, and thus by Fano's inequality we have
\begin{equation}\label{Fano1}
    H(Y^n|J, K)\leq n\eps_n,
\end{equation}
where $\eps_n\to 0$ as $n\to \infty$.

We start by obtaining a lower bound for $R_A$ as follows:
\begin{eqnarray*}
% \nonumber to remove numbering (before each equation)
  nR_A &\geq & H(J)\geq H(J|K) \\
   &=& H(Y^n, J|K)-H(Y^n|J, K)\\
   &\stackrel{(a)}{\geq}& H(Y^n, J|K)-n\eps_n\\
   &\geq& H(Y^n|K)-n\eps_n\\
   &=&\sum_{i=1}^n H(Y_i|Y^{i-1}, K)-n\eps_n\\
   &\geq &\sum_{i=1}^n H(Y_i|Y^{i-1}, X^{i-1}, K)-n\eps_n\\
   &\stackrel{(b)}{=}& \sum_{i=1}^n H(Y_i|V_i)-n\eps_n\\
   &\stackrel{(c)}{=}& H(Y_Q|V_Q, Q)-n\eps_n\\
   &\stackrel{(d)}{=}& nH(Y|V)-n\eps_n
\end{eqnarray*}
where $(a)$ follows from \eqref{Fano1}, and $(b)$ is due to the definition $V_i:=(Y^{i-1}, X^{i-1}, K)$. In $(c)$ we have introduced a time-sharing random variable $Q$ which is distributed uniformly over $\{1,2, \dots, n\}$ and is independent of $(X^n, Y^n, Z^n)$. In $(d)$ we have defined $V:=(V_Q, Q)$ and used the fact that $Y_Q$ has the distribution of $Y$ and hence we can replace $Y_Q$ with $Y$.

Next we obtain a lower bound on $R_C$:
\begin{eqnarray*}
% \nonumber to remove numbering (before each equation)
  nR_C &\geq & H(K)= I(Z^n; K)=\sum_{i=1}^n I(Z_i; K|Z^{i-1}) \\
   &\stackrel{(a)}{=}& \sum_{i=1}^n I(Z_i; K,Z^{i-1})\\
   &\stackrel{(b)}{=}& \sum_{i=1}^n I(Z_i; K,Z^{i-1}, X^{i-1}, Y^{i-1})\\
   &\geq & \sum_{i=1}^n I(Z_i; K,X^{i-1}, Y^{i-1})=nI(Z_Q; V_Q, Q)\\
   &=& nI(Z; V)
\end{eqnarray*}
where $(a)$ is due to the fact that $Z_i$ is independent of $Z^{i-1}$ for each $i$ and $(b)$ follows from the Markov chain relation $Z_i\markov (K, Z^{i-1})\markov (Y^{i-1}, X^{i-1})$.

We now upper bound the equivocation that any asymptotically lossless scheme produces. First we show the following identity which expresses $H(X^n|J)$ in terms of $H(Y^n|J)$ and some auxiliary terms:
\begin{equation}\label{Identity-intermediate}
% \nonumber to remove numbering (before each equation)
  H(X^n|J)-H(Y^n|J)=\sum_{i=1}^n[H(X_i|U_i)-H(Y_i|U_i)],
\end{equation}
where $U_i:=(X_{i+1}^n, Y^{i-1}, J)$. We will prove a general version of this identity later in Lemma~\ref{lemma_single_letter}.

%\newcounter{tempequationcounter}
\begin{figure*}[!t]
\normalsize
\setcounter{tempequationcounter}{\value{equation}}
\begin{IEEEeqnarray}{rCl}
\setcounter{equation}{5}
  0&\stackrel{(a)}{=}&\sum_{i=1}^nI(Y_i, E_i; X_{i+1}^n, E_{i+1}^n|J, Y^{i-1}, E^{i-1})-I(Y^{i-1}, E^{i-1}; X_i, E_i|J,X_{i+1}^n, E_{i+1}^n)   \nonumber\\
   &=& H(Y^n, E^n| J)-H(X^n,E^n| J)-\sum_{i=1}^n [H(Y_i,E_i|X_{i+1}^n, Y^{i-1}, E^{-i}, J)- H(X_i, E_i|X_{i+1}^n, Y^{i-1}, E^{-i}, J)]\nonumber\\
   &=& H(Y^n|E^n, J)-H(X^n|E^n, J)-\sum_{i=1}^n [H(Y_i|E_i, X_{i+1}^n, Y^{i-1}, E^{-i}, J)- H(X_i|E_i, X_{i+1}^n, Y^{i-1}, E^{-i}, J)]\nonumber\\
   &\stackrel{(b)}{=}&H(Y^n|E^n, J)-H(X^n|E^n, J)-\sum_{i=1}^n [H(Y_i|E_i, U_i)- H(X_i|E_i,U_i)]\label{eq:floatingequation_3}
\end{IEEEeqnarray}
\setcounter{equation}{\value{tempequationcounter}}
%\vspace*{4pt}
\hrulefill
\end{figure*}
The equivocation can be then be upper bounded as
\begin{eqnarray}
% \nonumber to remove numbering (before each equation)
 n(\Delta-\eps)&\leq & H(X^n|J) \nonumber\\
   &\stackrel{(a)}{=}& H(Y^n|J)+ \sum_{i=1}^n[H(X_i|U_i)-H(Y_i|U_i)] \nonumber\\
   &=& H(Y^n|K, J)+I(Y^n; K|J)\nonumber\\
   &&+\sum_{i=1}^n[H(X_i|U_i)-H(Y_i|U_i)] \nonumber\\
   &\leq & n\eps_n + I(K;Y^n, X^n|J)\nonumber\\
   &&+\sum_{i=1}^n[H(X_i|U_i)-H(Y_i|U_i)] \nonumber\\
%   &=& n\eps_n + H(K|J)-H(K|X^n, Y^n)\\
%   &&+\sum_{i=1}^n[H(X_i|U_i)-H(Y_i|U_i)]\\
   &\stackrel{(b)}{\leq} & n\eps_n + I(K;X^n, Y^n)\nonumber\\
   &&+\sum_{i=1}^n[H(X_i|U_i)-H(Y_i|U_i)]\nonumber\\
   &= & n\eps_n +\sum_{i=1}^nI(K; X_i, Y_i| X^{i-1}, Y^{i-1})\nonumber\\
   &&+ \sum_{i=1}^n[H(X_i|U_i)-H(Y_i|U_i)]\nonumber\\
   &= & n\eps_n +\sum_{i=1}^nI(K, X^{i-1}, Y^{i-1}; X_i, Y_i)\nonumber\\
   &&+ \sum_{i=1}^n[H(X_i|U_i)-H(Y_i|U_i)]\nonumber\\
   &=& n\eps_n + \sum_{i=1}^nI(V_i; X_i, Y_i)\nonumber\\
   &&+ \sum_{i=1}^n[H(X_i|U_i)-H(Y_i|U_i)]\nonumber\\
   &=& n\eps_n+nI(V_Q; X_Q, Y_Q|Q)\nonumber\\
   &&+n[H(X_Q|U_Q, Q)-H(Y_Q|U_Q, Q)]\nonumber\\
   &\stackrel{(c)}{=}& n\eps_n+nI(V_Q, Q; X_Q, Y_Q)\nonumber\\
   &&+n[H(X_Q|U_Q, Q)-H(Y_Q|U_Q, Q)]\nonumber\\
   &\stackrel{(d)}{=}& n\eps_n\nonumber\\&&+n[I(V; X, Y)+H(X|U)-H(Y|U)], \nonumber
\end{eqnarray}
where $(a)$ follows from \eqref{Identity-intermediate}, $(b)$ follows from the Markov chain relation $J\markov (X^n, Y^n)\markov K$ and hence $I(X^n, Y^n;K|J)\leq I(X^n, Y^n; K)$, $(c)$ is due to the fact that $Q$ is independent of $(X_Q,Y_Q)$ and in $(d)$ we have introduced $U:=(U_Q, Q)$.

We note that by definitions of $U$ and $V$, the Markov chain conditions $(X, Y)\markov Z\markov V$ and $Z\markov (X, Y)\markov U$ are satisfied. The cardinality bounds given in the statement of the theorem can be proved using support lemma \cite{csiszarbook}.
\end{proof}
\begin{remark}
As mentioned earlier, the special case $X=Y$ is studied in \cite{Secure_source_Tandon} where it is shown that for any achievable triple $(R_A, R_C, \Delta)$, the optimal equivocation satisfies $\Delta\leq I(Y; V)$. We see that Theorem~\ref{Theorem_lossless_Yamamoto} yields the same result and thus gives a tight bound in this special case.
\end{remark}
In practice, the private source $X$ might not be directly available to Alice. In this case, her mapping is $f_A:\Y^n\to \{1, 2, \dots, 2^{nR_A}\}$ and the above theorem reduces to the following corollary.
\begin{corollary}\label{Corollary_lossless_Yamamoto}
When the source $X^n$ is not available to Alice, any achievable triple $(R_A, R_C, \Delta)$ satisfies
\begin{eqnarray*}
% \nonumber to remove numbering (before each equation)
  R_A &\geq& H(Y|V), \\
  R_C &\geq & I(Z;V), \\
  \Delta &\leq & I(Y;V)+ H(X|U)-H(Y|U),
\end{eqnarray*}
for some $U\in \U$ and $V\in \V$ such that $P(x,y,z,u,v)=P(x,y,z)P(v|z)P(u|y)$ and $|\U|\leq |\Y|+1$ and $|\V|\leq |\Z|+2$.
\end{corollary}
\begin{proof}
The proof follows easily from the proof of Theorem~\ref{Theorem_lossless_Yamamoto}. In particular, introducing $V_i:=(Y^{i-1}, K)$ and $U_i:=(X_{i+1}^n, Y^{i-1}, J)$, we can follow easily the chain of inequalities given for the equivocation analysis with appropriate modifications. Since now $J=f_A(Y^n)$, we have $(X_i, Z_i)\markov Y_i\markov U_i$.
\end{proof}
%\begin{remark}
%When Alice does not see the private source, Corollary~\ref{Corollary_lossless_Yamamoto} gives an upper bound for the optimal equivocation which is related to $\max[H(X|U)-H(Y|U)]$ where maximization is taken over all $U$ such that $X\markov Y\markov U$. This quantity is equal to $H(X)$ if and only if $X$ and $Y$ are independent and is equal to zero if $X=Y$. We can, thus, view this quantity as an asymmetric measure of dependence between $X$ and $Y$.
%\end{remark}
%\begin{example}
%Suppose $Z^n=X^n$ and $R_C\geq H(X)$ and hence the private source $X^n$ is available to Bob. Then Theorem~\ref{Theorem_lossless_Yamamoto} implies that $R_A\geq H(Y|X)$ and $\Delta\leq H(X)+\max [H(X|U)-H(Y|U)]$ where the maximization is over random variable $U\in \U$ jointly distributed with $X$ and $Y$ with $|\U|\leq |\X|.|\Y|+1$. Since this maximization is always positive (indeed, the maximum value occurs at $U=Y$), one can conclude that perfect privacy is achieved. For example if $X$ and $Y$ are jointly Gaussian random variables, setting $J=f_A(X^n, Y^n)=Y^n-\E[Y^n|X^n]$, we have $J$ independent of $X^n$. On the other hand, in order for Bob to recover $Y^n$, he only needs to recover $\E[Y^n|X^n]$ given $X^n$ provided to him as side information.
%\end{example}
\section{Yamamoto's Lossless Source Coding: Uncoded Side Information at Eve}
We now turn our focus to the case where there is an eavesdropper, Eve, with perfect access to the channel from Alice to Bob and also side information $E^n$. Unlike in the last section, in this model the achievable $(R_A, R_C, \Delta)$ has not been fully characterized in the case of $X=Y$. However, G\"{u}nd\"{u}z et al.\ \cite{Secure_lossless_compression} and Probhakaran and Ramchandran \cite{Secure_distributed_Vinod} showed that if $R_C>H(Z)$, that is uncoded side information is available at Bob, then $(R_A, \Delta)$ is an achievable pair if and only if $R_A\geq H(Y|Z)$ and $\Delta\leq \max [I(Y;Z|U)-I(Y;E|U)]$ where the maximization is taken over $U$ that satisfies $Z\markov Y\markov U$, thus providing a full single-letter characterization of the achievable rate-equivocation region. In this section, we assume coded side information is available at Bob and Eve has uncoded side information $E^n$. As in \cite{Secure_source_Tandon}, we assume that the Eve's side information $E^n$ forms the Markov chain $X^n\markov Y^n\markov E^n$.

\subsection{A Converse Result}
We consider the model depicted in Fig.~\ref{Figue:Yamamoto_Lossless_eavesdropper} in which Eve has access to side information $E^n$ which satisfies $E^n\to Y^n\to X^n$.
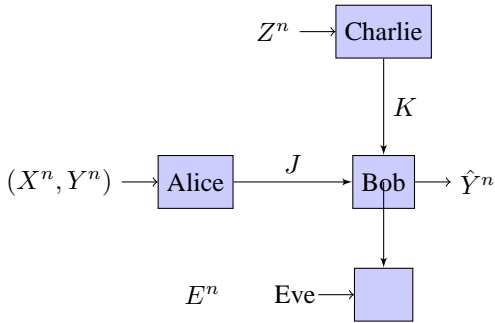
\begin{figure}[h]
\centering
\begin{tikzpicture}[node distance=2.5cm,auto,>=latex']
   % \node [int] (a) {Alice};
    \node [int, pin={[init]left:$(X^n,Y^n)$}] (a) [node distance=1.5cm] {Alice};
%    \node (b) [left of=a,node distance=1.5cm, coordinate] {a};
    \node [coordinate] (end) [right of=b, node distance=2cm]{};
    %\node (b) [left of=a,node distance=1.25cm, coordinate] {a};
    \node [int, pin={[init_to]right:$\hat{Y}^n$}] (c) [right of=a] {Bob};
    \node [int, pin={[init]left:$Z^n$}] (d) [above of=c, node distance=2cm] {Charlie};
    \hspace{-1.18cm}\node [int, pin={[init]left:$E^n$}] (e) [below of=c, node distance=1.5cm] {Eve};
    \hspace{1.18cm}
    \node [coordinate] (end) [right of=c, node distance=2cm]{};
    \node [coordinate] (Alice_channel) [above of=e, node distance=1.5cm]{};
    \path[->] (a) edge node {$J$} (c);
    \path[->] (d) edge node {$K$} (c);
    \hspace{-1.18cm}
    \draw[<-] (e) edge node {}(Alice_channel) ;
\end{tikzpicture}
\caption{Yamamoto's lossless source coding with eavesdropper having side information.}
\label{Figue:Yamamoto_Lossless_eavesdropper}
\end{figure}

The achievable $(R_A, R_C, \Delta)$ in this model is defined similarly as before with the utility constraint \eqref{error-probability-constraint} and the privacy constraint  \begin{equation}\label{privacy-constraint_eavesdropper}
    \addtocounter{equation}{1}
    \frac{1}{n}H(X^n|E^n, J)\geq \Delta-\eps.
\end{equation}
Before we get to an outer bound for the achievable region of this model, we need to state the following lemma which is a generalization of identity  \eqref{Identity-intermediate} that we used in the proof of Theorem~\ref{Theorem_lossless_Yamamoto}.
\begin{lemma}\label{lemma_single_letter}
Let $(J, X^n, Y^n, E^n)$ be jointly distributed according to $P(j, x^n, y^n, e^n)$. Then we can write:
\begin{eqnarray*}
% \nonumber to remove numbering (before each equation)
  H(X^n|E^n, J)&&\hspace{-0.85cm}-H(Y^n|E^n, J) \\
   &=&  \sum_{i=1}^n\left[H(X_i|E_i, U_i)-H(Y_i|E_i, U_i)\right]
\end{eqnarray*}
 where $U_i:=(X_{i+1}^n, Y^{i-1}, E^{-i}, J)$ for each $1\leq i\leq n$ and $E^{-i}:=(E_{i-1}, E_{i+1}^n)$ .
\end{lemma}
\begin{proof}
 The proof is presented in \eqref{eq:floatingequation_3}, where $(a)$ follows from Cisz\'{a}r sum identity \cite[page 25]{networkinfotheory}, in $(b)$ we used the definition of $U_i$.
\end{proof}
\begin{theorem}\label{Theorem_lossless_Yamamoto_Eavesdropper}
The set of all achievable triples $(R_A, R_C, \Delta)$ for this model when Eve is provided with side information $E^n$ and $E^n\markov Y^n\markov X^n$, satisfies
\begin{eqnarray*}
% \nonumber to remove numbering (before each equation)
  R_A &\geq & H(Y|V), \\
  R_C &\geq & I(Z; V),\\
  \Delta &\leq & I(X,Y; V)-I(X,Y; E|U)\\
  &&+H(X|E,U)-H(Y|E,U),
\end{eqnarray*}
for some $U$ and $V$ which form $(Z,E)\markov (X,Y)\markov U$ and $(X,Y,E)\markov Z\markov V$.
\end{theorem}
\begin{proof}
The lower bounds for both $R_A$ and $R_C$ follow along the same lines as in the proof of Theorem~\ref{Theorem_lossless_Yamamoto}. We shall show the upper bound for the equivocation. We note that since Bob is required to reconstruct $Y^n$ losslessly, Fano's inequality implies that
\begin{equation}\label{Fano_2}
    H(Y^n|J, K)\leq n\eps_n
\end{equation}
for $\eps_n\to 0$ as $n\to \infty$. As before, let $J=f_A(X^n, Y^n)$ and $K=f_C(Z^n)$.
%\begin{figure*}[!t]
%\normalsize
%\setcounter{tempequationcounter}{\value{equation}}
%\begin{IEEEeqnarray}{rCl}
%\setcounter{equation}{8}
%    \hspace{-0.7cm}H(X^n|E^n, J)&\stackrel{(a)}{=}& H(Y^n|E^n, J)+\sum_{i=1}^n [H(X_i|E_i, U_i)-H(Y_i|E_i, U_i)]\nonumber\\
%    &=& H(Y^n|J, Z^n)+I(Y^n; Z^n|J)-I(Y^n; E^n|J)+\sum_{i=1}^n [H(X_i|E_i, U_i)-H(Y_i|E_i, U_i)]\nonumber\\
%    &\stackrel{(b)}{\leq}& n\eps_n+I(X^n,Y^n; Z^n|J)-I(Y^n; E^n|J)+\sum_{i=1}^n [H(X_i|E_i, U_i)-H(Y_i|E_i, U_i)]\nonumber\\
%    &\stackrel{(c)}{\leq}& n\eps_n+I(X^n,Y^n; Z^n)-I(Z^n; J)-I(Y^n; E^n)+I(E^n; J)+\sum_{i=1}^n [H(X_i|E_i, U_i)-H(Y_i|E_i, U_i)]\nonumber\\
%    &\stackrel{(d)}{=}& n\eps_n+\sum_{i=1}^n [I(X_i,Y_i; Z_i)-I(Z_i; J, Z_{i+1}^n)-I(Y_i, X_i; E_i)\nonumber\\
%    &&+I(E_i; J, E^{i-1})+ H(X_i|E_i, U_i)-H(Y_i|E_i, U_i)]\nonumber\\
%    &\stackrel{(e)}{=}&n\eps_n+ \sum_{i=1}^n [I(X_i,Y_i; Z_i)-I(Z_i; J, Z_{i+1}^n, E^{i-1})-I(Y_i, X_i; E_i)+I(E_i; J, Z_{i+1}^n, E^{i-1})\nonumber\\
%    && + H(X_i|E_i, U_i)-H(Y_i|E_i, U_i)]\nonumber\\
%    &\stackrel{(f)}{\leq}&n\eps_n+\sum_{i=1}^n [I(X_i,Y_i; Z_i)-I(Z_i; U_i)-I(Y_i, X_i; E_i)+I(E_i; U_i)\nonumber\\
%    &&+ H(X_i|E_i, U_i)-H(Y_i|E_i, U_i)]\label{eq:floatingequation_5}
%\end{IEEEeqnarray}
%\setcounter{equation}{\value{tempequationcounter}}
%%\vspace*{4pt}
%\hrulefill
%\end{figure*}
%To upper bound the equivocation, we first need to establish  \eqref{eq:floatingequation_3}%
%\addtocounter{equation}{3} on the top of in which  \eqref{eq:floatingequation_2}%
%\addtocounter{equation}{3} follows from the Csisz\'{a}r's sum identity \cite[page 25]{networkinfotheory}.

The upper bound for the equivocation is obtained in \eqref{eq:floatingequation_4}%
\addtocounter{equation}{3} shown on top of the next page where $(a)$ follows from Lemma~\ref{lemma_single_letter} and $(b)$ is due to \eqref{Fano_2}. Since $K\markov (X^n, Y^n)\markov J$ and $E^n\markov Y^n\markov J$, we have $I(X^n, Y^n; K|J)\leq I(X^n, Y^n; K)$ and $I(Y^n; E^n|J)=I(Y^n; E^n)-I(E^n; J)$ and hence $(c)$ follows. We again used the Markov chain relation $E^n\markov Y^n\markov X^n$ in $(d)$. The definition $V_i:=(K, X^{i-1}, Y^{i-1})$ and the fact that $I(E_i; J, E^{i-1})\leq I(E_i; U_i)$ are used in $(e)$.  Note that since $U_i\markov (X_i, Y_i)\markov E_i$ we have in $(f)$ that $I(X_i, Y_i; E_i|U_i)=I(X_i, Y_i; E_i)-I(E_i; U_i)$. The proof completes by introduction of a time sharing random variable $Q$ uniformly distributed over $\{1, 2,\dots, n\}$ and independent of $(X^n, Y^n, Z^n, E^n)$ and letting $X=X_Q, Y=Y_Q$, $E=E_Q$, $V=(V_Q, Q)$ and $U=(U_Q,Q)$.
\end{proof}
\begin{figure*}[!t]
\normalsize
\setcounter{tempequationcounter}{\value{equation}}
\begin{IEEEeqnarray}{rCl}
\setcounter{equation}{8}
    H(X^n|E^n, J)&\stackrel{(a)}{=}& H(Y^n|E^n, J)+\sum_{i=1}^n [H(X_i|E_i, U_i)-H(Y_i|E_i, U_i)]\nonumber\\
    &=& H(Y^n|J, K)+I(Y^n; K|J)-I(Y^n; E^n|J)+\sum_{i=1}^n [H(X_i|E_i, U_i)-H(Y_i|E_i, U_i)]\nonumber\\
    &\stackrel{(b)}{\leq}& n\eps_n+I(X^n,Y^n; K|J)-I(Y^n; E^n|J)+\sum_{i=1}^n [H(X_i|E_i, U_i)-H(Y_i|E_i, U_i)]\nonumber\\
    &\stackrel{(c)}{\leq}& n\eps_n+I(X^n,Y^n; K)-I(Y^n; E^n)+I(E^n; J)+\sum_{i=1}^n [H(X_i|E_i, U_i)-H(Y_i|E_i, U_i)]\nonumber\\
    &\stackrel{(d)}{=}& n\eps_n+\sum_{i=1}^n [I(X_i,Y_i; K, X^{i-1}, Y^{i-1})-I(Y_i, X_i; E_i)+I(E_i; J, E^{i-1})\nonumber\\
    &&+ H(X_i|E_i, U_i)-H(Y_i|E_i, U_i)]\nonumber\\
    &\stackrel{(e)}{\leq}&n\eps_n+\sum_{i=1}^n [I(X_i,Y_i; V_i)-I(Y_i, X_i; E_i)+I(E_i; U_i)+ H(X_i|E_i, U_i)-H(Y_i|E_i, U_i)]\nonumber\\
    &\stackrel{(f)}{=}&n\eps_n+\sum_{i=1}^n [I(X_i,Y_i; V_i)-I(Y_i, X_i; E_i|U_i)+ H(X_i|E_i, U_i)-H(Y_i|E_i, U_i)]\nonumber\\
    &\stackrel{(g)}{=}& n\eps_n+ I(X_Q,Y_Q; V_Q, Q)-I(Y_Q, X_Q; E_Q|U_Q, Q)+ H(X_Q|E_Q, U_Q, Q)-H(Y_Q|E_Q, U_Q, Q)]\label{eq:floatingequation_4}
\end{IEEEeqnarray}
\setcounter{equation}{\value{tempequationcounter}}
%\vspace*{4pt}
\hrulefill
\end{figure*}
\begin{remark}
Setting $E^n=\emptyset$ and thus removing the eavesdropper's side information,  Theorem~\ref{Theorem_lossless_Yamamoto_Eavesdropper} yields $\Delta\leq I(X,Y;V)+H(X|U)-H(Y|U)$ and hence Theorem~\ref{Theorem_lossless_Yamamoto_Eavesdropper} subsumes Theorem~\ref{Theorem_lossless_Yamamoto}.
\end{remark}
In the simple case of $X=Y$, the optimal scheme when coded side information is available at Bob and $E^n=\emptyset$ is proposed in \cite{Secure_source_Tandon} which is shown to resemble  the binning scheme of Wyner in \cite{Wyner_source_coidng_wit_SI}. Although, a tight bound for the equivocation when $E^n$ is available is not yet known, Theorem~\ref{Theorem_lossless_Yamamoto_Eavesdropper}, specialized to $X=Y$, implies
$$\Delta\leq I(Y; V)-I(Y; E|U),$$
for auxiliary random variables $U$ and $V$ which form Markov chains $V\markov Z\markov (Y, E)$ and $U\markov Y\markov (Z, E)$.

\subsection{A Coding Scheme When Bob Has Uncoded Side Information}
As a special case, we consider the case where Alice does not see the private source and also $R_C>H(Z)$ (i.e., Bob has uncoded side information). In this case, Theorem~\ref{Theorem_lossless_Yamamoto_Eavesdropper} implies that the best achievable equivocation is upper bounded by $$\max[I(Y; Z)-I(Y; E|U)+H(X|E, U)-H(Y|E, U)],$$ where the maximization is taken over $U$ which forms the Markov chain relation $U\markov Y\markov (Z, E, X)$. In the following we give a simple coding scheme which incurs a smaller equivocation and is thus suboptimal. In fact, if the above maximization results in a $U$ which is independent of $Z$, then the following coding scheme is optimal. On the other hand, if the maximization results in a $U$ which is constant, then it implies that Slepian-Wolf binning is optimal, because if Alice uses Slepian-Wolf binning then the equivocation is equal to $H(X|E)-H(Y|Z)$, as observed in \cite{Secure_distributed_Vinod}.
\begin{theorem}\label{Theorem_lossless_Yamamoto_TightUnCoded_Eavesdropper}
When $X^n$ is not given to Alice and Bob observes side information $Z^n$,  then $(R_A, \Delta)$ which satisfies
\begin{eqnarray*}
% \nonumber to remove numbering (before each equation)
  R_A &\geq & H(Y|Z), \\
  \Delta &\leq & I(Y; Z|U)-I(Y; E|U)\\
  &&+H(X|E,U)-H(Y|E,U),
\end{eqnarray*}
is achievable where the auxiliary random variable $U$ forms the Markov chain $(X, Z, E)\markov\ Y\markov U$.
\end{theorem}
\begin{proof}
Our scheme is similar to the ones proposed in \cite{Secure_lossless_compression} and \cite{curt_serecy_Cheap}.
 Given $Y^n$, we generate $2^{n(I(Y; U)+\eps)}$ independent codewords of length $n$, $U^n(w)$, $w\in \{1, 2, \dots, 2^{nI(Y; U)+\eps}\}$ according to $\prod_{i=1}^nP(u_i)$. We then uniformly bin all the $U^n$ sequences into $2^{n(I(Y; U)-I(U; Z))}$ bins. Let $B(i)$ be the indices assigned to bin $i$. There are approximately $2^{nI(U;Z)}$ indices in each bin. We also uniformly bin $Y^n$ sequences into $2^{n(H(Y|U,Z)+\eps)}$ bins and let $C(k)$ be the set of sequences $Y^n$ in bin $k$.
\begin{figure*}[!t]
\normalsize
\setcounter{tempequationcounter}{\value{equation}}
\begin{IEEEeqnarray}{rCl}
\setcounter{equation}{9}
% \nonumber to remove numbering (before each equation)
  H(X^n|U^n, E^n)&=&\sum_{(u^n, e^n)\in \U^n\times \mathcal{E}^n}P(u^n, e^n)H(X^n|U^n=u^n, E^n=e^n) \nonumber\\
  &\geq& \sum_{(u^n, e^n)\in \mathcal{T}^n_{U,E}}P(u^n, e^n)H(X^n|U^n=u^n, E^n=e^n) \nonumber\\
  &=&   \sum_{(u^n, e^n)\in \mathcal{T}^n_{U,E}}P(u^n, e^n)\left[-\sum_{x^n\in \X^n} P(x^n|u^n, e^n)\log (P(x^n|u^n, e^n))\right]\nonumber\\
   &\geq& \sum_{(u^n, e^n)\in \mathcal{T}^n_{U,E}}P(u^n, e^n)\left[-\sum_{x^n\in \T^n_{X|u^n, e^n}} P(x^n|u^n, e^n)\log (P(x^n|u^n, e^n))\right]\nonumber\\
   &\stackrel{(c)}{\geq}& n (H(Y|U, E)-\delta_n) \sum_{(u^n, e^n)\in \mathcal{T}^n_{U,E}}P(u^n, e^n)\left[\sum_{x^n\in \T^n_{X|u^n, e^n}}P(x^n|u^n, e^n)\right] \nonumber\\
   &=& n (H(Y|U, E)-\delta_n) \sum_{(u^n, e^n)\in \mathcal{T}^n_{U,E}}P(u^n, e^n) \left[\Pr\{(u^n, e^n, X^n)\in \T^n_{X|u^n, e^n}\}\right] \nonumber\\
   &\stackrel{(d)}{\geq}& n (H(Y|U, E)-\delta_n)(1-\delta'_n)\label{eq:floatingequation_7}
\end{IEEEeqnarray}
\setcounter{equation}{\value{tempequationcounter}}
\hrulefill
\end{figure*}
Alice adopts a two-part encoding scheme. Given $Y^n$, Alice, in the first part, looks for a codeword $U^n(w)$ such that $(Y^n, U^n(w))\in \A_{YU}^n$, where $\A_{YU}^n$ denotes the set of all strongly typical $(y^n, u^n)\in \Y^n\times \U^n$ with respect to the distribution $P(y,u)$. She then reveals the bin index $J_1$ such that $w\in B(J_1)$. In the second part, she reveals $J_2$ such that $Y^n\in C(J_2)$.

Given $J_1$, $J_2$ and $Z^n$, Bob can find, with high probability, $U^n(w)$ such that $w\in B(J_1)$ and $(U^n(w), Z^n)\in \A_{ZU}^n$. It is then clear from the Slepian-Wolf theorem that Bob can recover $Y^n$ with high probability given $U^n(w)$, $Z^n$, and $J_2$.

The rate of this encoder is clearly equal to $H(Y|U,Z)+I(Y; U)-I(U;Z)=H(Y|Z)$.

The equivocation for this scheme can be found as
\begin{eqnarray*}
% \nonumber to remove numbering (before each equation)
 && H(X^n|J_1, J_2, E^n) \\ &= &H(X^n|J_1, E^n)-I(X^n; J_2|J_1, E^n) \\
   &\geq & H(X^n|U^n, E^n)-H(J_2) \\
   &\stackrel{(a)}{\geq} & H(X^n|U^n, E^n)-nH(Y|U,Z)\\
   &\stackrel{(b)}{\geq} & n[H(X|U, E)-H(Y|U,Z)]\\
   &=& n[H(X|E, U)-H(Y|E, U)\\
   && +I(Y; Z|U)-I(Y; E|U)],
\end{eqnarray*}
where $(a)$ follow from the fact that $J_2$ is a random variable over a set of size $2^{nH(Y|U,Z)}$ and $(b)$ is proved in  \eqref{eq:floatingequation_7} where $\mathcal{T}^n_{U,E}$ denotes the set of typical sequences $(u^n, e^n)$ and  $(c)$ is due to the property of typical sequences; in particular for  typical $x^n$ sequence with respect to $P(x^n|u^n, e^n)$ for $(u^n, e^n)\in \mathcal{T}^n_{U,E}$ we have $P(x^n| u^n, e^n)\leq 2^{-(n(H(X|U,E)-\delta(n)))}$ for $\delta_n\to 0$ as $n\to \infty$. We invoked Markov lemma \cite[Lemma 12.1]{networkinfotheory} in $(d)$ to conclude that for the Markov chain relation $(X,E)\markov Y\markov U$ we have $(x^n, y^n, e^n, u^n)\in\mathcal{T}^n_{X,Y,E,U}$ and hence $\Pr\{(u^n, e^n, X^n)\in \T^n_{U,E, X}\}>1-\delta'_n$ for each pair $(u^n, e^n)\in (u^n, e^n)\in \mathcal{T}^n_{U,E}$ and $\delta'_n\to 0$ as $n\to \infty$.
\end{proof}

\section{Concluding Remarks}
Having combined the idea of compression of private and public sources of Yamamoto \cite{yamamotoequivocationdistortion} with secure source coding problem (e.g. \cite{Secure_lossless_compression}, \cite{Secure_source_Tandon} and \cite{Secure_distributed_Vinod}), we introduced a lossless source coding problem in which, given a two-dimensional source $(X^n, Y^n)$, the encoder must compress the source into an index $J$ with rate $R_A$ such that the receiver recovers $Y^n$ losslessly and simultaneously reveals only little information about $X^n$. This model differs from typical information-theoretic secrecy models in that the utility and privacy constraints are defined for two different sources and thus provides a more general utility-equivocation tradeoff.

We gave converse results for compression rates and also the information leakage rate (or equivocation)  which reduce to known results in the special case of $X=Y$. In particular, with this simplifying assumption, Theorem~\ref{Theorem_lossless_Yamamoto} and Theorem~\ref{Theorem_lossless_Yamamoto_TightUnCoded_Eavesdropper} reduce to \cite[Theorem 1]{Secure_source_Tandon} and \cite[Corollary 3.2]{Secure_lossless_compression}.

However, it is not clear at the moment that the bounds are tight in general. Constructing an achievability scheme for the most general case (i.e., the setting of Theorem~\ref{Theorem_lossless_Yamamoto_Eavesdropper}) is the subject of our future studies.

\vspace{-0.05cm}
\bibliographystyle{IEEEtran}
\bibliography{bibliography}

% Generated by IEEEtran.bst, version: 1.14 (2015/08/26)
\begin{thebibliography}{10}
\providecommand{\url}[1]{#1}
\csname url@samestyle\endcsname
\providecommand{\newblock}{\relax}
\providecommand{\bibinfo}[2]{#2}
\providecommand{\BIBentrySTDinterwordspacing}{\spaceskip=0pt\relax}
\providecommand{\BIBentryALTinterwordstretchfactor}{4}
\providecommand{\BIBentryALTinterwordspacing}{\spaceskip=\fontdimen2\font plus
\BIBentryALTinterwordstretchfactor\fontdimen3\font minus
  \fontdimen4\font\relax}
\providecommand{\BIBforeignlanguage}[2]{{%
\expandafter\ifx\csname l@#1\endcsname\relax
\typeout{** WARNING: IEEEtran.bst: No hyphenation pattern has been}%
\typeout{** loaded for the language `#1'. Using the pattern for}%
\typeout{** the default language instead.}%
\else
\language=\csname l@#1\endcsname
\fi
#2}}
\providecommand{\BIBdecl}{\relax}
\BIBdecl

\bibitem{yamamotoequivocationdistortion}
H.~Yamamoto, ``A source coding problem for sources with additional outputs to
  keep secret from the receiver or wiretappers,'' \emph{IEEE Trans. Inf.
  Theory}, vol.~29, no.~6, pp. 918--923, Nov. 1983.

\bibitem{Secure_distributed_Vinod}
V.~Prabhakaran and K.~Ramchandran, ``On secure distributed source coding,'' in
  \emph{IEEE Inf. Theory Workshop (ITW)}, Sept. 2007, pp. 442--447.

\bibitem{Secure_lossless_compression}
D.~G\"{u}nd\"{u}z, E.~Erkip, and H.~Poor, ``Secure lossless compression with
  side information,'' in \emph{Proc. IEEE Inf. Theory Workshop}, May 2008, pp.
  169--173.

\bibitem{Secrure_lossless_Gunduz_ISIT}
------, ``Lossless compression with security constraints,'' in \emph{IEEE Int.
  Sym. on Inf. Theory (ISIT)}, July 2008, pp. 111--115.

\bibitem{Secure_Source_Piantanida}
J.~Villard and P.~Piantanida, ``Secure multiterminal source coding with side
  information at the eavesdropper,'' \emph{IEEE Trans. Inf. Theory,}, vol.~59,
  no.~6, pp. 3668--3692, June 2013.

\bibitem{Secure_source_Tandon}
R.~Tandon, S.~Ulukus, and K.~Ramchandran, ``Secure source coding with a
  helper,'' \emph{IEEE Trans. Inf. Theory}, vol.~59, no.~4, pp. 2178--2187,
  April 2013.

\bibitem{Wyner_source_coidng_wit_SI}
A.~Wyner, ``On source coding with side information at the decoder,'' \emph{IEEE
  Trans. Inf. Theory,}, vol.~21, no.~3, pp. 294--300, May 1975.

\bibitem{Secure_lossy_coding}
E.~Ekrem and S.~Ulukus, ``Secure lossy source coding with side information,''
  in \emph{Proc. Annual Allerton Conference on Communication, Control, and
  Computing}, Sept. 2011, pp. 1098--1105.

\bibitem{csiszarbook}
I.~Csisz\'{a}r and J.~K\"{o}rner, \emph{Information Theory: Coding Theorems for
  Discrete Memoryless Systems}.\hskip 1em plus 0.5em minus 0.4em\relax
  Cambridge University Press, 2011.

\bibitem{networkinfotheory}
Y.~H. Kim and A.~E. Gamal, \emph{Network Information Theory}.\hskip 1em plus
  0.5em minus 0.4em\relax Cambrdige University press, 2012.

\bibitem{curt_serecy_Cheap}
C.~Schieler and P.~Cuff, ``Secrecy is cheap if the adversary must
  reconstruct,'' in \emph{Proc. IEEE Int. Symp. on Inf. Theory (ISIT),}, July
  2012, pp. 66--70.

\end{thebibliography}
\end{document}